\documentclass{article}

\usepackage[preprint,nonatbib]{neurips_2020}

\usepackage{amsthm}
\usepackage[utf8]{inputenc} 
\usepackage[T1]{fontenc}    
\usepackage{url}            
\usepackage{booktabs}       
\usepackage{amsfonts}       
\usepackage{nicefrac}       
\usepackage{microtype}      
\usepackage{amsmath}
\usepackage{amsthm}
\usepackage{xcolor}
\newtheorem{theorem}{Theorem}
\newtheorem{lemma}[theorem]{Lemma}
\newtheorem{claim}[theorem]{Claim}
\usepackage[ruled,vlined]{algorithm2e}
\DeclareMathOperator*{\argmax}{argmax}
\DeclareMathOperator*{\argmin}{argmin}

\newtheorem{defn}{Definition}

\title{Exploiting No-Regret Algorithms in System Design}

\author{
  Le Cong Dinh\\
  University of Southampton\\
  United Kingdom\\
  \texttt{l.c.dinh@soton.ac.uk} \\
  \And
  Nick Bishop \\
  University of Southampton\\
  United Kingdom\\
  \texttt{nb8g13@soton.ac.uk} \\
  \And
  Long Tran-Thanh \\
  University of Warwick\\
  United Kingdom\\
  \texttt{long.tran-thanh@warwick.ac.uk} \\
}

\begin{document}

\maketitle

\begin{abstract}
We investigate a repeated two-player zero-sum game setting where the column player is also a designer of the system, and has full control on the design of the payoff matrix. In addition,  the row player uses a no-regret algorithm to efficiently learn how to adapt their strategy to the column player's behaviour over time in order to achieve good total payoff. 
The goal of the column player is to guide her opponent to pick a mixed strategy which is favourable for the system designer. Therefore, she needs to: (i) design an appropriate payoff matrix $A$ whose unique minimax solution contains the desired mixed strategy of the row player; and (ii) strategically interact with the row player during a sequence of plays in order to guide her opponent to converge to that desired behaviour. 
To design such a payoff matrix, we propose a novel solution that provably has a unique minimax solution with the desired behaviour.
We also investigate a relaxation of this problem where uniqueness is not required, but all the minimax solutions has the same mixed strategy for the row player.
Finally, we propose a new game playing algorithm for the system designer and prove that it can guide the row player, who may play a \emph{stable} no-regret algorithm, to converge to a minimax solution.
 
\end{abstract}

\section{Introduction}

We consider a repeated two-player zero-sum game setting, in which one player (say the column player) is also a designer of the system (i.e., she can design the payoff matrix), and her opponent (the row player) is a strategic unility maximiser, who can efficiently learn how to adapt their strategy to the column player's behaviour over time in order to achieve good total payoff. 
 The goal of the column player is to guide her opponent to pick up a mixed strategy which is favourable for the system designer. In particular, she needs to achieve this by: (i) designing an appropriate payoff matrix $A$; and (ii) strategically interacting with the row player during a sequence of plays in order to guide her opponent to converge to the desired behaviour.
 
 This problem is motivated by a number of real-world applications.
 For example, consider a shared-resource system in which selfish and rational agents independently decide how much resources they want to use from the pool of shared resources in order to increase their own utility (the more they get the better for them).
Being selfish, the agents will prefer maximising their own utility over optimising the social welfare, which often leads to the depletion of the shared resources. This situation is known as the tragedy of commons in the economics and computer science literature. One way to mitigate this situation is to have a designated agent (a.k.a. the system designer) who has the power to reshape the payoff of the other agents by interfering into the system whenever the others take an action~\cite{han2018cost,hilbe2018evolution,tilman2020evolutionary}. By repeatedly doing so over time, the system designer's goal is to enforce the other players, guided by their strategic and selfish nature, to converge to a desired behaviour which can help maintain the amount of the shared resources in the long run.

Another example comes from the domain of security games~\cite{shi2018designing,sinha2018stackelberg,xu2016playing}. 
In this setting, it is typical to have  a defender and an attacker repeatedly playing a zero-sum two-player game with each other. 
The attacker's goal is to cause as much damage as possible by attacking a set of targets at each time step, and the defender's goal is to protect these targets from the attacker.
A key research question in this domain is how the defender can prevent the attacker to choose the most important targets by faking the value of such attacks (e.g., by using honeypots to fool cyber attackers in computer networks, or to provide false information of animals' whereabout  in the fight against illegal poachers). Once this is done, the next step is then how the defender should interact with the attacker over time so that the attacker believes the false targets are in fact the important ones.

While this problem has a wide applicability across many domains of game theory and computer science  (as shown by the examples above), it is also difficult to solve.
In fact, it consists of the following two sub-problems, each of which is a challenging problem themselves:

\par \textbf{Games with a unique minimax solution}. It is well known that in two-player zero-sum games, if both players are perfectly strategic, then the best they can achieve is the minimax solution of the game~\cite{Nicolo06}.
This naturally lends itself to the idea of designing a game whose minimax solution contains the desired behaviour of the row player.
In this game, due to their rational nature, the row player will follow the desired behaviour.
However, if the game has more than one minimax solution, the row player might pick a solution which differs from the desired behaviour. 
Therefore, to guarantee that the row player will choose the desired behaviour for sure, the system designer has to ensure that the game has a unique minimax solution with the desired behaviour of the row player as part of the solution~\footnote{This requirement can be relaxed by allowing multiple minimax solutions, but each of them should have the same behaviour for the row player. See Theorem~\ref{theorem: unique for row player} for more details.}. 
While this is an age-old problem and can be dated back to the seminal work of Shapley, Karlin, and Bohnenblust~\cite{bohnenblust1950}, it was largely ignored by the research community and only started regaining attention since the recent work of~\cite{kuleshov2015inverse}. 

\par \textbf{Last round convergence in zero-sum games}.
Once a game with unique minimax solution is established, the next step is to incentivise the row player to play the minimax solution.
The easiest way is perhaps to share the information about the minimax solution with the row player, and assuming that they are rational, we can hope that they will eventually pick up the desired behaviour.
However, this only works if there is a trust between the players (i.e., the row player trusts the information provided by the system designer), which might not always be the case in non-cooperative settings~\cite{xu2015exploring,dughmi2019algorithmic}.
In addition, it can be easily shown that by repeatedly playing her part of the minimax strategy, the system designer will not be able to push the row player to converge either~\footnote{We will discuss this in more detail in Section~\ref{section: last round convergence}.}. 

So how can we still be able to enforce the row player to play the desired behaviour?
A potential answer lies within the well known result that if the players use no-regret algorithms (see Section~\ref{section: problem setting} for the definition of this term) to choose their strategies over time, the game dynamic will eventually converge to the minimax solution of the game on average. This can be done without requiring any information exchange between the players, as they only need to use the observed payoffs to update their strategies~\cite{Nicolo06,Freund99}. 
Therefore, our hypothesis is that if the row player uses a no-regret algorithm to learn the minimax solution, we can somehow exploit this and guide the row player towards the desired behaviour.
However, convergence on average does not necessarily imply convergence in the actual chosen strategies, and thus, cannot guarantee the actual convergence in behaviour~\cite{mertikopoulos2018}. This issue is known as the last round convergence problem in the online learning literature, and has recently attracted a good amount of attention~\cite{Daskalakis2018c,daskalakis2017training,dinh2020last}.  

\subsection{Main Contributions}

Against this background, this paper aims to address the abovementioned research questions as follows:

\par 1. We first propose a simple and computationally efficient approach to design a two-player zero-sum game in which the minimax strategy for the row player is unique, and is the desired behaviour (Theorems~\ref{First theorem of constructing matrix A} and \ref{theorem: unique for row player}). The efficiency of our algorithm lies within the fact that we only require the minimax strategy of the row player to be predefined, whilst the minimax strategies of the column player are free to be chosen as required. Note that other existing techniques aim to solve the case where the behavior of both players are predefined. 
In addition, our result directly exploits the connection between linear programming (LP) and zero-sum games \cite{goldmanLP}, which will hopefully shed new light on the literature of finding games with unique minimax solutions.


\par 2. In addition, we prove that LRCA, an algorithm recently proposed by~\cite{dinh2020last}, can be used by the system designer to guide the row player to achieve last round convergence, even if the row player chooses from \emph{a large class} of no-regret algorithms to play (Theorem~\ref{extension of LRCA algorithm}). In particular, this class of algorithms has a \emph{stability} property (see Section~\ref{section: last round convergence} for more details). We also prove that a wide class of popular no-regret algorithms, Follow the Regularized Leader (FTRL), has this property (Theorem`\ref{FTRL has stability property}).

\subsection{Related Work}
\label{section: related work}

Our problem setting can be viewed from as a mechanism design (MD) problem~\cite{nisanAGT}. From this perspective, we can consider the game matrix chosen by the designer as the mechanism, and the actions chosen by each participant as the information they choose to report. 
%
In this domain, perhaps the most similar framework to our problem setting is the online MD framework~\cite{parkesAGT,gerding2011online}, in which a central mechanism must make decisions over time as different agents arrive and depart at different time steps. However, our setting deals with agents which do not depart or arrive, but rather gain knowledge about the central mechanism as time moves on. 
Secondly, the goal of the designer is distinct from typical MD settings. Rather than standard solution concepts such as incentive compatibility or social welfare, we aim for the goal of guiding players into playing specific strategies. Such solution concepts are common amongst the online learning community in which the problem of playing a repeated game against another agent is explored under various conditions. 

The problem of constructing zero-sum games with unique minimax equilibrium was first considered by~\cite{bohnenblust1950}. In fact ~\cite{bohnenblust1950} provide an algorithm for constructing a zero-sum game with a chosen unique minimax equilibrium. In contrast, we provide an alternate and computationally efficient method for finding a zero-sum game with a given unique minimax equilibrium, by directly exploiting the connection between linear programming (LP) and zero-sum games~\cite{goldmanLP}.
In particular, the topic of certifying the uniqueness of linear programming solutions has been studied extensively within the optimisation community. Appa~\cite{appa2002uniqueness} provides a constructive method for verifying the uniqueness of an LP solution which requires solving an addition linear program. Mangasarian~\cite{mangasarian1979uniqueness} describes a number of conditions which guarantee the uniqueness of an LP solution. In fact, it is one of these conditions that we shall leverage to derive our methods for constructing zero-sum games with unique solutions. More generally, many other works deal with characterising the optimal solution sets of linear programs \cite{Kantor, Tavas}, but do not typically pay any special attention to uniqueness. 

Finally, while convergence on average of no-regret learning dynamics has been studied extensively in game theory and online learning communities (e.g.,~\cite{Nicolo06,Freund99}), last round convergence has only been a topic of research in the last few years.  This started with the negative result of Bailey and Piliouras \cite{Bailey2018}, who showed that if agents use a no-regret multiplicative weight update algorithm then the last round strategy converges to the boundary (i.e., it will not converge to any minimax solution). Later, Mertikopoulos \emph{et al.}~\cite{mertikopoulos2018} showed that there is a loop in the strategies played if both agents follow regularized learning algorithms. Daskalakis and Panageas~\cite{Daskalakis2018c} showed that, if both players play according to the optimistic multiplicative weight update (OMWU) algorithm, then they will achieve last round convergence. More recently, Dinh \emph{et al.}~\cite{dinh2020last} proposed the last round convergence in asymmetric games (LRCA) algorithm, which if played by one agent, guarantees last round convergence for both agents if the other agent follows a certain no-regret algorithm. 

\section{Preliminaries}
To begin, we shall introduce some basic definitions from game theory through which our problem setting will be formally described. We define a finite normal form two-player zero-sum game, $\Gamma$, by a tuple $(\mathcal{N}, \mathcal{A}, u)$. We denote the set of players by $\mathcal{N} = \{1, 2\}$. Each player $i \in \mathcal{N}$ is required to select from a finite action set $\mathcal{A}_{i}$ simultaneously. We denote by $\mathcal{A} = \mathcal{A}_{1} \times  \mathcal{A}_{2}$ the set of all possible combinations of actions that may be chosen by the players. After actions have been chosen, each player $i$ is rewarded according to their payoff function $u_{i}: \mathcal{A} \rightarrow \mathbb{R}$. 
In the zero-sum games setting, we have $u_{2} = -u_{1}$. Such games can be succinctly represent by a single matrix $A \in \mathbb{R}^{|\mathcal{A}_{1}| \times |\mathcal{A}_{2}|}$, where entry $A_{ij}$ denotes the payoff of player 1 when actions $i \in \mathcal{A}_{1}$ and $j \in \mathcal{A}_{2}$ are selected by the players. 
Additionally, players are allowed to randomise their action selection. Formally speaking, player $i$ can select any probability distribution $s \in \Delta(\mathcal{A}_{i})$ over their action set. An action is then selected by randomly sampling according to this distribution. We refer to this set of probability distributions as the strategies available to the player. We say that a strategy is pure if it corresponds to deterministically choosing a single action, otherwise we say that a strategy is mixed.
%
%

\par Given this matrix representation, we will interchangeably refer to player $1$ and player $2$ as the row and column players respectively. For notational convenience, we shall also refer to $|\mathcal{A}_{1}|$ by $n$ and $|\mathcal{A}_{2}|$ by $m$. We will denote the strategy adopted by player 1 (row player) by the vector $\mathbf{x}$, where $\mathbf{x}_{i}$ refers to the probability of player 1 choosing action $i$. Similarly, we will use the vector $\mathbf{y}$ for for the column player's adopted strategy.
Note that the expected payoff for player 1 when strategies $(\mathbf{x}, \mathbf{y})$ are played is given by $-\mathbf{x}^{T}A\mathbf{y}$. Similarly the expected payoff for player 2 is simply $\mathbf{x}^{T}A\mathbf{y}$. It is well known that, for any payoff matrix $A$, there is a minimum expected payoff that both players can guarantee, known as the value of the game, $v \in \mathbb{R}$. The existence of $v$ is a direct consequence of Von Neumann's Minimax theorem: $
    \label{minmax}
    v = \min_{\mathbf{x}\in\Delta_{n}}\max_{\mathbf{y} \in \Delta_{m}}\mathbf{x}^{T}A\mathbf{y}
    = \max_{\mathbf{y}\in\Delta_{m}}\min_{\mathbf{x} \in \Delta_{n}}\mathbf{x}^{T}A\mathbf{y}$. 

Any strategy which is guaranteed to attain expected payoff $v$ for a given player is called a minimax strategy. Any pair $(\mathbf{x}^{*}, \mathbf{y}^{*})$ of minimax strategies are said to form a minimax equlibrium. Finding a minimax strategy for a given player is equivalent to solving one of the saddle point problems in equation \ref{minmax}, which is in turn equivalent to solving a linear programming (LP) problem. For example, finding a minimax strategy $\mathbf{x}^{*}$ for the row player corresponds to solving the following LP problem:

\begin{equation}
\label{xlp}
\begin{aligned}
    \min_{v, \; \mathbf{x}}v \quad
    \text{s.t. \quad} v - \sum_{i=1}^{n}A_{ij}\mathbf{x}_{i} \geq 0 \quad \forall j \in [m], \quad 
    \sum_{i=1}^{n} \mathbf{x}_{i} = 1, \quad
    \mathbf{x}_{i} \geq 0 \quad \forall i \in [n]
\end{aligned}
\end{equation}


\section{Problem Setting}
\label{section: problem setting}
\par In this paper, we consider a setting in which two players are engaged in a repeated zero-sum game, $\Gamma$, for $T$ time steps. At each time step, $1 \leq t \leq T$, each player must commit to a mixed strategy $\mathbf{x}_{t}$ and $\mathbf{y}_{t}$. 
We will refer to a single element of such a vectors by $\mathbf{x}_{i, t}$ and $\mathbf{y}_{i,t}$, respectively.
After choosing their mixed strategies, each player observes their payoff vectors, $A\mathbf{y}_{t}$ and $A^{T}\mathbf{x}_{t}$, which correspond to the expected payoff of each of player's pure strategies against the strategy chosen by the other player.


\par The row player has no initial knowledge of $\Gamma$ and only has its history of previous strategy choices and payoff vectors, $((\mathbf{x}_{1}, A\mathbf{y}_{1}), \dots, (\mathbf{x}_{t-1}, A\mathbf{y}_{t-1}))$,  to inform its selection of $\mathbf{y}_{t}$. The goal of the row player is to minimise their average payoff 
  $\frac{1}{T}\sum_{t=1}^{T}\mathbf{x}_{t}^{T}A\mathbf{y}_{t}$. 
%
In contrast, as well as being a participant in the game, the column player also takes the role of the system designer. Therefore, the column player has full knowledge of $\Gamma$ prior to playing. The column player is tasked with guiding the row player into selecting a strategy $\mathbf{x}^{*}$ at time step $T$, while simultaneously maximising their own average payoff. To aid in this endeavour, the column player is permitted to choose a payoff matrix $A$ prior to the start of the game. Thus the problems of the column player are two-fold in the following sense:

\begin{enumerate}
    \item The column player must choose matrix $A$ which allow them to guide the row player reliably to the desired final iterate $\mathbf{x}^{*}$.
    \item Given $A$, the column player must select a sequence of mixed strategies which guide the row player to the desired final iterate while still accumulating high average payoff.
\end{enumerate}

\par We shall now formalise these goals. Firstly, we shall introduce the well-established notion of regret as a metric for measuring the performance of both players with respect to the payoffs they accumulate over time. 

\begin{defn}
The regret of any sequence of strategies $(\mathbf{x}_{1}, \dots, \mathbf{x}_{T})$ chosen by the row player with respect to a fixed strategy $\mathbf{x}$ is given by
\begin{equation}
    \mathcal{R}_{T, \mathbf{x}} = \sum_{t=1}^{T}\mathbf{x}^{T}_{t}A\mathbf{y}_{t} - \sum_{t=1}^{T}\mathbf{x}^{T}A\mathbf{y}_{t}
\end{equation}
that is, the regret is the difference between the payoff accumulated by the sequence $(\mathbf{x}_{1}, \dots, \mathbf{x}_{t})$ and the payoff accumulated by the sequence where a given fixed mixed strategy $\mathbf{x}$ is chosen at each time step. A similar notion of regret is defined for the column player.
\end{defn}

\par We shall say that the row player is `successful` in accumulating high average payoff if their regret with respect to the sequence of strategies chosen by the column player is sublinear in $T$:

\begin{equation}
    \lim_{T \rightarrow \infty}\:\max_{\mathbf{x} \in \Delta_{n}}\:\frac{\mathcal{R}_{T, \mathbf{x}}}{T} = 0
\end{equation}

\par A notion of successfully accumulating high average payoff is similarly defined for the column player.
%
%
%
Lastly, we must define what it means for the column player to be successful in guiding the row player into selecting a desired final iterate. We say that the column player is successful in guiding the row player into selecting a chosen strategy $\mathbf{x}^{*}$ if the last iterate strategy of the row player is within a $\delta$-neighbourhood of $\mathbf{x}^{*}$.

\begin{defn}
 $\forall \delta > 0$: a mixed strategy $\mathbf{x}$ is in a $\delta$-neighbourhood of $\mathbf{x}^{*}$ if and only if $\|\mathbf{x} - \mathbf{x}^{*}\|_{2} \leq \delta$.
\end{defn}





\par Next, we shall describe our approach for ensuring that the column player is simultaneously successful in achieving sublinear regret and guiding the column player into an appropriate final iterate. Firstly, we shall choose $A$ so that $\mathbf{x}^{*}$ is the only minimax strategy available to the row player. In the section that follows, we will describe a computationally cheap method for finding such a matrix.

\par In order to guide the row player, the column player will play strategies according to an algorithm which guarantees last iterate convergence to minmax equilibria in the repeated zero-sum setting, where the row player follows a no-regret algorithm to choose her strategy at each time step. As the only minmax strategy available to the row player is $\mathbf{x}^{*}$, the last iterate will inevitably converge to $\mathbf{x}^{*}$. In particular, the column player will select mixed strategies according to the LRCA algorithm, however any algorithm which guarantees last iterate convergence to minimax equilibria would suffice. In Section~\ref{section: last round convergence} we provide a detailed description of LRCA. In addition, to guaranteeing last iterate convergence to minimax equilibria, LRCA also achieves sublinear regret with respect to any sequence played by the row player, thus also satisfying the criteria of the column player with respect to accumulating high payoff.



\section{Designing Games with a Unique Minimax Solution}

\par We will now describe how we may construct a matrix $A$ which describes a zero-sum game with unique minimax strategy $\mathbf{x}^{*}$ for the row player. Firstly, we select a minimax strategy, $\mathbf{y}^{*}$, for the column player. The only constraint we place on $\mathbf{y}^{*}$ is that its support must be greater than equal to the support of $\mathbf{x}^{*}$. For if this is not the case, then it is provably impossible to construct a matrix $A$ with minimax equilibrium $(\mathbf{x}^{*}, \mathbf{y}^{*})$ whilst guaranteeing the uniqueness of the minimax strategy $\mathbf{x}^{*}$ \cite{bohnenblust1950}. From now on, without loss of generality, we shall assume that for any strategy with support $k$, that the first $k$ entries are nonzero.

\par Once $\mathbf{y}^{*}$ has been selected, we proceed directly with the construction of the matrix $A$. If the support of $\mathbf{y}^{*}$ is equal to the support of $\mathbf{x}^{*}$, then $A$ is constructed in accordance with Theorem \ref{First theorem of constructing matrix A}. On the other hand, if the support of $\mathbf{y}^{*}$ is greater than the support of $\mathbf{x}^{*}$, then we construct $A$ according to Theorem \ref{theorem: unique for row player}. Note that, in both theorems, $\mathbf{y}^{*}$ is not necessarily the unique minimax strategy for the column player, whilst $\mathbf{x}^{*}$ is the unique minimax strategy for the row player.

The core idea behind both theorems is to construct three systems of inequalities using linear programming theory. The satisfaction of the first two of these three systems guarantees that $\mathbf{x}^{*}$ and $\mathbf{y}^{*}$ are minimax strategies for a given matrix $A$. Meanwhile, the infeasibility of the third system guarantees the uniqueness of $\mathbf{x}^{*}$ with respect to a given matrix $A$. Thus, the task of finding a matrix $A$ with the desired properties corresponds to finding a matrix which satisfies the first two systems whilst guaranteeing infeasibility of the third.

The first system is derived by expressing the problem of finding minimax equilibria for a given matrix $A$ as the LP (\ref{xlp}). Note that the solution set of LP (\ref{xlp}) consists of the pairs $(v, \mathbf{x})$ which satisfy the Karush-Kuhn-Tucker (KKT) system for LP (\ref{xlp}), that is, the $(v, \mathbf{x})$ which satisfy the following system of inequalities:
%

%
%
\begin{equation}
\label{ineqs}
    \begin{aligned}
        &v - \sum_{i=1}^{n}A_{ij}\mathbf{x}_{i} \geq 0  \quad &\forall j \in [m] \\
        &\sum_{i=1}^{n} \mathbf{x}_{i} = 1, \quad 
        \mathbf{x}_{i} \geq 0  \quad &\forall i \in [n]\\
        &\mathbf{y}_{j}^{*}(v - \sum_{i=1}^{n}A_{ij}\mathbf{x}_{i}) = 0 \quad &\forall j \in [m]
    \end{aligned}
\end{equation}

\par Thus, to ensure that the row player has a minimax strategy $\mathbf{x}^{*}$, we must choose a matrix $A$ such that $\mathbf{x}^{*}$ satisfies the system of inequalities (\ref{ineqs}). Assuming that $\mathbf{x}^{*}$ is a feasible mixed strategy, the inequalities which pertain to the matrix $A$ are therefore:

\begin{equation}
\label{feas1}
\begin{aligned}
    v \geq \sum_{i=1}^{n}A_{ij}\mathbf{x}^{*}_{i}  \quad \forall j \in [m], \quad \quad 
    v = \sum_{i=1}^{n}A_{ij}\mathbf{x}^{*}_{i} \quad \forall j \text{ where } \mathbf{y}^{*}_{j} > 0
\end{aligned}
\end{equation}



The second system of inequalities is constructed in a similar fashion, instead using the version of LP (\ref{xlp}) corresponding to the column player:
\begin{equation}
\label{feas2}
\begin{aligned}
    v \leq \sum_{i=1}^{n}A_{ij}\mathbf{x}^{*}_{i}  \quad \forall i \in [n], \quad \quad 
    v = \sum_{j=1}^{m}A_{ij}\mathbf{y}_{j}^{*} \quad \forall i \text{ where } \mathbf{x}^{*}_{i} > 0
\end{aligned}
\end{equation}

Thus, if a matrix $A$ satisfies this second system of inequalities, then $\mathbf{y}^{*}$ is a minimax strategy for the column player. In contrast, the third system of inequalities is constructed by leveraging the following theorem concerning uniqueness of LP solutions:

\begin{theorem} [Theorem 2 from ~\cite{mangasarian1979uniqueness}]
\label{uniquethm}
Let $\mathbf{x}^{*}$ be a solution to the following LP:
\begin{equation}
    \label{thmeqn}
    \begin{aligned}
    \min_{\mathbf{x}} \:\mathbf{p}^{T}\mathbf{x} \quad \quad 
    \text{s.t. \quad} A\mathbf{x} = \mathbf{b}, \quad 
    C\mathbf{x} \geq \mathbf{d}
    \end{aligned}
\end{equation}
then $\mathbf{x}^{*}$ is a unique solution if and only if there exists no $\mathbf{x}$ satisfying
   $A\mathbf{x} = 0$, $C_{J}\mathbf{x} \geq 0$, $\mathbf{p}^{T}\mathbf{x} \leq 0$, $\mathbf{x} \neq 0$,
where $J$ is the index set of tight inequality constraints, i.e.,
  $J = \{i \:|\:C_{i}\mathbf{x}^{*} = \mathbf{d}_{i}\}$.
\end{theorem}
 
\par Assume, for the sake of analysis, that $(v, \mathbf{x}^{*})$ is a solution of LP (\ref{xlp}). Then the corresponding version of system (\ref{thmeqn}) for LP (\ref{xlp}) is:


\begin{equation}
\label{uniquesystem}
\begin{aligned}
&\mathbf{x}_{i} \geq 0 \quad \forall i \text{ where } \mathbf{x}^{*}_{i} = 0, \quad
\hat{v} - \sum^{n}_{i=1}A_{ij}\mathbf{x}_{i} \geq 0 \quad \forall j \text{ where } \mathbf{y}^{*}_{j} > 0, \\
&\sum_{i=1}^{n}\mathbf{x}_{i} = 0,\quad 
\hat{v} \leq 0, \quad
(\hat{v}, \mathbf{x}) \neq 0
\end{aligned}
\end{equation}

\par where $\mathbf{x}$ and $\hat{v}$ are variables. According to Theorem \ref{uniquethm}, if a chosen matrix $A$ guarantees the infeasibility of system (\ref{uniquesystem}) whilst satisfying system (\ref{feas1}) and the corresponding system for the column player, then $\mathbf{x}^{*}$ must be the unique minimax strategy for the row player. We now can prove the following lemma:

\begin{lemma}
\label{lem}
Let $A$ be a matrix with the a minimax solutions $(\mathbf{x}^{*}, \mathbf{y}^{*})$. Then $\mathbf{x}^{*}$ is the unique minimax strategy for the row player if, for every $k \in n$ such that $\mathbf{x}^{*}_{k} > 0$ then there exists a column $j$ of $A$ such that:
\begin{equation*}
    A_{ij} =
    \begin{cases}
        \alpha& \mathbf{x}^{*}_{i} = 0  \\
        \gamma & i = k \\
        \beta & \text{otherwise}
    \end{cases}
\end{equation*}
where $\alpha > \gamma > 0$ and $\beta > \gamma$.
\end{lemma}

\begin{proof}
Let $A$ be a matrix as described above. Since $(\mathbf{x}^{*}, \mathbf{y}^{*})$ is a minimax solution, proving that $\mathbf{x}^{*}$ is the unique minimax strategy for the row player corresponds to showing that system (\ref{uniquesystem}) is infeasible given the matrix $A$.

Firstly, observe that if there is a solution to system (\ref{uniquesystem}) with $\hat{v} < 0$, then there must also be a solution  with $\hat{v} = 0$. Thus, without loss of generality, we can assume that $\hat{v} = 0$ and rewrite the system as follows:
\begin{equation}
\label{simpleunique}
\begin{aligned}
& \mathbf{x}_{i} \geq 0 \quad \forall i \text{ where } \mathbf{x}^{*}_{i} = 0, \quad
\sum^{n}_{i=1}A_{ij}\mathbf{x}_{i} \leq 0 \quad \forall j \text{ where } \mathbf{y}^{*}_{j} > 0,
\sum_{i=1}^{n}\mathbf{x}_{i} = 0, \quad 
\mathbf{x} \neq 0
\end{aligned}
\end{equation}
Now, consider a candidate solution,  $\mathbf{x} \in \mathbb{R}^{n}$, to system (\ref{uniquesystem}). Such a vector must be nonzero, and its elements must sum to zero. In addition, we know that for any $i$ such that $\mathbf{x}^{*}_{i} = 0$, that $\mathbf{x}_{i} \geq 0$. Since $\mathbf{x}$ is not the zero vector and its elements sum to zero, there must be one element of $\mathbf{x}$ which is negative, $\mathbf{x}_{k} < 0$. Furthermore, note that $\mathbf{x}_{k}^{*}$ must be positive.  

Rewriting the second and third conditions of system (\ref{simpleunique}) in term of $\mathbf{x}_{k}$, we have:
\begin{equation*}
    \begin{aligned}
        &\mathbf{x}_{k}= -\sum_{i\neq k}^{n}\mathbf{x}_{i}, \quad
        \sum^{n}_{i \neq k}A_{ij}\mathbf{x}_{i} \leq -A_{kj}\mathbf{x}_{k} \: \forall j \text{ where } \mathbf{y}^{*}_{j} > 0
    \end{aligned}
\end{equation*}
 Using the equality, we can rewrite each inequality as follows:
\begin{equation*}
    \sum^{n}_{i \neq k}A_{ij}\mathbf{x}_{i} \leq A_{kj}\sum^{n}_{i \neq k}\mathbf{x}_{i}, \:\forall j  \text{ where }  \mathbf{y}^{*}_{j} > 0.
\end{equation*}
Dividing both sides by $\sum^{n}_{i \neq k}\mathbf{x}_{i}$, we obtain the following set of inequalities:
\begin{equation}
    \label{violate}
    \sum^{n}_{i \neq k}A_{ij}\mathbf{x}^{\prime}_{i} \leq A_{kj} \quad \forall j \text{ where } \mathbf{y}^{*}_{j} > 0
\end{equation}
where
    $\mathbf{x}_{i}^{\prime} = \frac{\mathbf{x}_{i}}{\sum_{i \neq k}^{n}\mathbf{x}_{i}}$ 
is the normalisation of the entry $\mathbf{x}_{i}$ by the sum of all entries, excluding $\mathbf{x}_{k}$. 
 Moreover, let 
\begin{equation*}
  \tilde{x} = \sum_{i \neq k, \: \mathbf{x}_{i}^{*} = 0}\mathbf{x}^{\prime}_{i}
\end{equation*}
denote the summation over normalised entries $\mathbf{x}^{\prime}_{i}$ for which $\mathbf{x}^{*}_{i} = 0$, again excluding $\mathbf{x}^{\prime}_k$. As $\mathbf{x}_{k} < 0$, we must have $\sum^{n}_{i \neq k}\mathbf{x}_{i} > 0$. Therefore, $\mathbf{x}_{i}^{\prime} \geq 0$ for all $i$ where $\mathbf{x}^{*}_{i} = 0$, and thus, $\tilde{x}$ is nonnegative and less than 1. Note that, by definition, there is a column $j$ of $A$ such that
\begin{equation*}
 \sum^{n}_{i \neq k}A_{ij}\mathbf{x}^{\prime}_{i} = \tilde{x}\alpha + (1 - \tilde{x})\beta > \tilde{x} \gamma + (1 - \tilde{x})\gamma = \gamma  = A_{kj}
\end{equation*}
However, this contradicts the inequalities (\ref{violate}), and thus, $\mathbf{x}$ cannot be a solution of system (\ref{simpleunique}). Therefore, system (\ref{simpleunique}) has no solutions, implying that system (\ref{uniquesystem}) also has no solutions. Hence, $\mathbf{x}^{*}$ is the unique minimax strategy for the row player.
\end{proof}

\begin{theorem}\label{First theorem of constructing matrix A}
Let $\mathbf{x} \in \Delta_n$, $\mathbf{y} \in \Delta_m$ such that $support(\mathbf{x})=support(\mathbf{y})=k$. Define matrix $A$ as:
\[A=\begin{bmatrix}
    a_1 & \alpha_2 & ... & \alpha_k & v &... & v \\
    \alpha_1 & a_2 & ... & \alpha_k & v &... & v \\
    ...&...&...&...&...&...&... \\
    \alpha_1 &\alpha_2 &...& a_k & v &... & v \\
    \alpha_1-z &\alpha_2 -z &...& \alpha_k-z & v &... & v \\
    ...&...&...&...&...&...&... \\
    \alpha_1-z &\alpha_2 -z &...& \alpha_k-z & v &... & v
    
  \end{bmatrix}
\]
where $v >0$ and
$0 < z < \min_{i \in [k]}(v, \frac{v\mathbf{y}_i}{1-\mathbf{x}_i}); \; \alpha_i = v+ z \frac{\mathbf{x}_i}{\mathbf{y}_i}; \; a_i= v - z \frac{1-\mathbf{x}_i}{\mathbf{y}_i}\;\forall i \in [k]$.
Then $\mathbf{x}$ is the unique minimax equilibrium of the row player in the zero-sum game described by the matrix $A$.
\end{theorem}
\begin{proof}
We will prove this theorem by considering two cases. In one case, we will assume that $m=n$ and $support(\mathbf{x})=support(\mathbf{y})=n$. In the other case, we will assume that  $support(\mathbf{x})=support(\mathbf{y})=k$ with $k \leq \min(n,m)$. Proving that the theorem holds in both cases yields the result.

\par We first prove the theorem in the case where $n=m$ and $support(\mathbf{x})=support(\mathbf{y})=n$. Consider the following matrix:
\[
A=
  \begin{bmatrix}
    a_1 & \alpha_2 & ... & \alpha_n \\
    \alpha_1 & a_2 & ... & \alpha_n \\
    ...&...&...&... \\
    \alpha_1 &\alpha_2 &...& a_n
  \end{bmatrix}
\]
By Lemma \ref{lem}, we know that the row player has a unique minimax strategy if $a_{i} < \alpha_{i}$ for all $i \in [n]$. Thus if $\alpha_{i}$ and $a_{i}$ are chosen so that $(\mathbf{x}, \mathbf{y})$ is a minimax equilibrium, whilst respecting the aforementioned conditions, then the proof is complete. Ensuring that $(\mathbf{x}, \mathbf{y})$ is a minimax equilibrium is equivalent to ensuring that $A$ satisfies systems (\ref{feas1}) and (\ref{feas2}). 

Since the support in this case is maximal, system (\ref{feas1}) simplifies to:
\begin{equation}\label{proof 1 Nip2020 1st equation}
 \mathbf{x}^TA= \left[v,v,\dots, v\right]   
\end{equation}
and system (\ref{feas2}) simplifies to:
\begin{equation}\label{proof 1 Nip2020 2nd equation}
A\mathbf{y}= \left[v,v,\dots, v\right]^T
\end{equation}
Note that equality (\ref{proof 1 Nip2020 1st equation})  is equivalent to
\begin{equation}
    \mathbf{x}_1(a_1-\alpha_1)+\alpha_1= \mathbf{x}_2(a_2-\alpha_2)+\alpha_2= \dots = \mathbf{x}_n(a_n-\alpha_n)+\alpha_n=v \label{proof 1 Nip2020 3nd equation}.
\end{equation}
Similarly, observe that equality (\ref{proof 1 Nip2020 2nd equation}) is equivalent to
\begin{equation}
    \mathbf{y}_1(\alpha_1-a_1) = \mathbf{y}_2(\alpha_2-a_2)=\dots =\mathbf{y}_n(\alpha_n-a_n)=z. \label{proof 1 Nip2020 4th equation}  
\end{equation}
where $z = v - (\alpha_{1}\mathbf{y}_{1} + \dots + \alpha_{n}\mathbf{y}_{n})$.
Substituting the equalities (\ref{proof 1 Nip2020 4th equation}) into the equalities (\ref{proof 1 Nip2020 3nd equation}), we have the following formulae for each $\alpha_{i}$:
\[\alpha_i = v+ z \frac{\mathbf{x}_i}{\mathbf{y}_i} \quad \forall i \in [n].
\]
By selecting $z >0$ we guarantee that $\alpha_{i} > a_{i}$ for all $i \in [n]$ (a direct consequence of the inequalities (\ref{proof 1 Nip2020 4th equation})). Moreover, selecting $z$ as follows guarantees $a_{i} > 0$ and $\alpha_{i} > 0$ for all $i \in [n]$:
\[
z < \min_{i \in [k]}\left(v, \frac{v\mathbf{y}_i}{1-\mathbf{x}_i}\right)\]

By back substituting each $\alpha_{i}$ into the equalities (\ref{proof 1 Nip2020 4th equation}), we obtain the following formulae for each $a_{i}$:
\[a_i= \alpha_i+\frac{v-\alpha_i}{\mathbf{x}_i} = v-z\frac{1-\mathbf{x}_i}{\mathbf{y}_i} >0  \quad \forall i \in [n]
\]

 Since every $a_{i}$ and $\alpha_{i}$ has been chosen to satisfy the conditions of Lemma \ref{lem}, whilst also guaranteeing that $(\mathbf{x}, \mathbf{y})$ is a minimax solution for the matrix $A$, it must be the case that $\mathbf{x}$ is the unique minimax strategy for the row player. Therefore, in the case $support(\mathbf{x})=support(\mathbf{y})=n=m$, the matrix A with the entries as described in the theorem statement will have $\mathbf{x}$ as the unique minimax strategy for the row player.

Next, we consider the case in which we have $support(\mathbf{x})=support(\mathbf{y})=k$ with $k \leq \min(n,m)$. By swapping around the rows and columns of matrix A, without loss of generality, we can assume that $\mathbf{x}$ and $\mathbf{y}$ have the following forms:
\begin{equation*}
    \begin{aligned}
    \mathbf{x}=\left[\mathbf{x}_1, \mathbf{x}_2,...,\mathbf{x}_k, 0, 0, ..., 0\right]^T \text{ where } \mathbf{x}_i >0 \quad \forall i \in [k] \\
    \mathbf{y}=\left[\mathbf{y}_1, \mathbf{y}_2,...,\mathbf{y}_k, 0, 0, ..., 0\right]^T \text{ where } \mathbf{y}_i > 0 \quad \forall i \in [k]
    \end{aligned}
\end{equation*}
With $\mathbf{x}'=\left[\mathbf{x}_1, \mathbf{x}_2,...,\mathbf{x}_k\right]^T,\; \mathbf{y}'=\left[\mathbf{y}_1, \mathbf{y}_2,...,\mathbf{y}_k\right]^T$, following the above method, we can construct a matrix $A'$ for which $\mathbf{x}'$ is the unique minimax strategy for the row player:
 \[A'=\begin{bmatrix}
    a_1 & \alpha_2 & ... & \alpha_k \\
    \alpha_1 & a_2 & ... & \alpha_k \\
    ...&...&...&... \\
    \alpha_1 &\alpha_2 &...& a_k
  \end{bmatrix}
\]
Now, we can construct the matrix A as follows:
\[A=\begin{bmatrix}
    a_1 & \alpha_2 & ... & \alpha_k & v &... & v \\
    \alpha_1 & a_2 & ... & \alpha_k & v &... & v \\
    ...&...&...&...&...&...&... \\
    \alpha_1 &\alpha_2 &...& a_k & v &... & v \\
    \alpha_1-z &\alpha_2 -z &...& \alpha_k-z & v &... & v \\
    ...&...&...&...&...&...&... \\
    \alpha_1-z &\alpha_2 -z &...& \alpha_k-z & v &... & v
    
  \end{bmatrix}
\]
Firstly, note that the value of $A$ is the same as the value of $A'$. Secondly, note that $\alpha_{i} - z > a_{i}$ for all $i \in [k]$, and thus, the first $k$ columns of $A$ satisfy the conditions described in Lemma \ref{lem}.
Lastly, note that
\begin{equation*}
   \begin{aligned}
   \mathbf{x}^TA= \left[v,v,\dots v\right], \; A\mathbf{y}=\left[v,v,\dots v\right]^T.   
   \end{aligned} 
\end{equation*}
Therefore, $(\mathbf{x}, \mathbf{y})$ is a minimax solution for the matrix $A$ and, according to Lemma $\ref{lem}$, $\mathbf{x}$ is therefore the unique minimax strategy for the row player.
\end{proof}
In Theorem \ref{First theorem of constructing matrix A}, if the support of $\mathbf{y}$ is equal to $m$, then $\mathbf{y}$ is the unique minimax strategy for the column player \cite{bohnenblust1950}.
\begin{theorem}
\label{theorem: unique for row player}
Let $\mathbf{x} \in \Delta_n$, $\mathbf{y} \in \Delta_m$ such that $k = support(\mathbf{x}) < l=support(\mathbf{y})$. Let the matrix $A$ be of the form
\[
A=
  \begin{bmatrix}
    a_1 & \alpha_2 & ... & \alpha_k & \beta_1& ... &\beta_1 \\
    \alpha_1 & a_2 & ... & \alpha_k & \beta_2& ... &\beta_2 \\
    ...&...&...&... \\
    \alpha_1 &\alpha_2 &...& a_k & \beta_k& ... &\beta_k\\
    \alpha_1-z &\alpha_2-z &...& \alpha_k-z & v& ... &v\\
    ...&...&...&... \\
    \alpha_1-z &\alpha_2-z &...& \alpha_k-z & v& ... &v
    
  \end{bmatrix}
\]
where the parameters of $A$ satisfy
\begin{equation*}
    \begin{aligned}
     &0< v_1 < v\bar{y}, \quad \quad \bar{y}= \sum_{i=k+1}^l \mathbf{y}_i,\quad z=\frac{v \bar{y}-v_1}{\sum_{i=1}^k \mathbf{y}_i}, \\
     &\beta_i= v, \quad 
     \alpha_i= v+\frac{x_i (v\bar{y}-v_1)}{\mathbf{y}_i},\quad
     a_i=\alpha_i - \frac{v\bar{y}-v_1}{\mathbf{y}_i} \;\forall i \in [k]. 
    \end{aligned}
\end{equation*}
then $\mathbf{x}$ is the unique minimax strategy for the row player in the zero-sum game described by $A$.
\end{theorem}
\begin{proof}
Once again, we prove the theorem case by case.

\par First, suppose $support(\mathbf{x})=n$ and $support(\mathbf{y})=m$ i.e.\ both $\mathbf{x}$ and $\mathbf{y}$ are fully-mixed. Consider the matrix of the following form:
\[
A=
  \begin{bmatrix}
    a_1 & \alpha_2 & ... & \alpha_n & \beta_1& ... &\beta_1 \\
    \alpha_1 & a_2 & ... & \alpha_n & \beta_2& ... &\beta_2 \\
    ...&...&...&... \\
    \alpha_1 &\alpha_2 &...& a_n & \beta_n& ... &\beta_n
  \end{bmatrix}
.\]
If $a_i < \alpha_i$ for all $i \in k$ then, by Lemma \ref{lem}, the row player has a unique minimax strategy. Thus, similarly to Theorem \ref{First theorem of constructing matrix A}, we simply need to choose $a_{i}$, $\alpha_{i}$ and $\beta_{i}$ so that $(\mathbf{x}, \mathbf{y})$ is a minimax strategy of $A$, whilst respecting the aforementioned conditions on $a_{i}$ and $\alpha_{i}$.
\par Using systems (\ref{feas1}) and (\ref{feas2}), and the same logic applied in Theorem \ref{First theorem of constructing matrix A}, we know that $(\mathbf{x}, \mathbf{y})$ is  a minimax strategy of $A$ if the following equalities hold:
\begin{equation}\label{xeq}
    \mathbf{x}_1(a_1-\alpha_1)+\alpha_1= \mathbf{x}_2(a_2-\alpha_2)+\alpha_2= \dots = \mathbf{x}_n(a_n-\alpha_n)+\alpha_n=\sum_{i=1}^n \mathbf{x}_i \beta_i =v 
\end{equation}
\begin{equation}\label{yeq}
        (a_i-\alpha_i)\mathbf{y}_i+\sum_{j=1}^n \alpha_j \mathbf{y}_j +\beta_i \bar{y} =v \quad \forall i \in [n] \text{ where }  \bar{y}= \sum_{i=n+1}^m \mathbf{y}_i
\end{equation}
Setting $v_{1} = v - \sum^{n}_{j=1}\alpha_{j}\mathbf{y}_{j}$, we note that (\ref{yeq}) can be written as:
\begin{equation*}
    (a_i-\alpha_i)\mathbf{y}_i+\beta_i \bar{y} = v_1 \quad \forall i \in [n]
\end{equation*}
Furthermore, noting that $(a_i - \alpha_i) = (v_1 - \beta_{i}\bar{y})\; /\; \mathbf{y}_{i}$, equation (\ref{xeq}) can be written as:
\begin{equation*}
    \frac{\mathbf{x}_i}{\mathbf{y}_i}(v_1-\beta_i \bar{y}) +\alpha_i = \sum^{n}_{i=1}\mathbf{x}_{i}\beta_{i}=v \quad \forall i \in [n]
\end{equation*}
Thus, the system of equations we must satisfy to guarantee that $(\mathbf{x}, \mathbf{y})$ is a minimax equilibirum of the matrix $A$ is:
\begin{align*}
         (a_i-\alpha_i)\mathbf{y}_i+\beta_i \bar{y} &= v_1 \quad \forall i \in [n] \\
         \frac{\mathbf{x}_i}{\mathbf{y}_i}(v_1-\beta_i \bar{y}) +\alpha_i &= v  \quad \forall i \in [n] \\
         \sum^{n}_{i=1}\mathbf{x}_{i}\beta_{i}&=v
\end{align*}
To get rid of the third equality, we simply set $\beta_i = v$, for all $i \in [n]$. The system of equations then becomes:
\begin{align*}
    (a_i-\alpha_i)\mathbf{y}_i+v \bar{y} &= v_1 \quad \forall i \in [n] \\
    \frac{\mathbf{x}_i}{\mathbf{y}_i}(v_1-v\bar{y}) +\alpha_i &= v  \quad \forall i \in [n] \\
\end{align*}
and hence we see that:
\begin{align*}
    \alpha_{i} &= v + \frac{\mathbf{x}_i}{\mathbf{y}_i}(v\bar{y} - v_1) \quad \forall i \in [n] \\
    a_i &= \alpha_i - \frac{v\bar{y} - v_1}{\mathbf{y}_{i}} \quad \forall i \in [n]
\end{align*}
Thus, the parameters as described in the theorem statement guarantee that $(\mathbf{x}, \mathbf{y})$ is a minimax equilibrium of $A$. To guarantee that $0 < a_{i} < \alpha_{i}$ it suffices to choose $0 < v_1 < v\bar{y}$. Thus the parameters described in the theorem statement guarantee that all the conditions of Lemma \ref{lem} hold. As a result, $\mathbf{x}$ must be the unique minimax strategy for the row player when $\mathbf{x}$ and $\mathbf{y}$ are fully-mixed.

We can apply procedure similar to that in the proof of Theorem \ref{First theorem of constructing matrix A} to extend this argument to the case where $\mathbf{x}$ and $\mathbf{y}$ are not fully-mixed.
\end{proof}

\section{Last Round Convergence in Two-Player Zero-Sum Games}
\label{section: last round convergence}

Given the construction of a matrix $A$ with the desired property in the previous section, we now turn to investigate how to guide the row player to achieve last round convergence when the row player uses a no-regret algorithm. 
%
To do so, we first show that a na\"{i}ve approach, namely to repeatedly playing $\mathbf{y}^*$, will not always lead to the desired last round convergence (i.e., the row player will converge to $\mathbf{x}^{*}$): 
\begin{claim}
\label{claim: naive strategy doesn not work}
If $support(\mathbf{x}) > 1$
, then there is no guarantee that if the column player repeatedly plays $\mathbf{y}^*$, the row player will eventually converge to $\mathbf{x}^{*}$. 
\end{claim}

\begin{proof}
Since $support(\mathbf{x}^{*}) > 1$, the mixed strategy $\mathbf{x}$ places positive probability mass on at least two different pure strategies. It is easy to prove that any distribution over these two pure strategies also achieves the value of the game against $\mathbf{y}^*$. Consider any no-regret algorithm whose internal state depends only on the rewards observed, such as the multiplicative weight update algorithm (MWU). In the case of MWU, as the rewards observed for both pure strategies are the same, the weights assigned to each pure strategy will be the same at every time step, and thus at every time step both pure strategies will be played with equal probability. Therefore, if the minimax strategy of the row player is not a uniform distribution over the pure strategies in its own support, then the row player will not converge to their minimax strategy. Therefore, unless $support(\mathbf{x}) = 1$ (which is very rare in real world applications), by just repeatedly playing $\mathbf{y}^*$ the column player cannot guide the row player to converge to $\mathbf{x}^{*}$.
\end{proof}
%
Given this result, 
we need to design a different game playing policy for the column player.
As mentioned in Section~\ref{section: related work}, Dinh \emph{et al.}~\cite{dinh2020last} have proposed a new algorithm called LRCA, which if played by the column player (i.e., the system designer), can lead to last round convergence (for both agents) if the other agent plays one of the following popular no-regret algorithms: (i) multiplicative weight update; (ii) online mirror descent; and (iii) linear multiplicative weight update.
%
%
The pseudo code of LRCA is described in Algorithm~\ref{LRCA algorithm 1} below:

\begin{algorithm}
\textbf{Input:} Current iteration $t$, past feedback $x_{t-1}^\top A$ of the row player\\
\textbf{Output:} Strategy $y_t$ for the column player\\
\If{$t=2k-1, \;   k \in \mathbb{N}$}{$y_t=y^*$}
\If{$t=2k, \; k \in \mathbb{N}$}{ $e_t : = \argmax_{e \in \{e_1,e_2,\dots e_m\}}x_{t-1}^\top Ae$;$ \quad f(x_{t-1}): = \max_{y \in \Delta_m}x_{t-1}^\top Ay$ \\
$\alpha_t: = \frac{f(x_{t-1})-v}{\max \left(\frac{n}{4},2\right)}$\\
$y_t: =(1-\alpha_t)y^*+\alpha_t e_t$} 
\caption{\emph{L}ast \emph{R}ound \emph{C}onvergence with \emph{A}symmetry (LRCA) algorithm}\label{LRCA algorithm 1}
\end{algorithm}

We prove that LRCA can also work with other no-regret algorithms besides those  listed above, should the no-regret algorithm of the row player have a ``stability" property defined below:
\begin{defn}\label{stability condition}
A no-regret algorithm is \emph{stable} if $\; \forall t: \mathbf{y}_t=\mathbf{y}^* \implies \mathbf{x}_{t+1}=\mathbf{x}_t$.
\end{defn} 
\begin{theorem}\label{extension of LRCA algorithm}
Assume that the row player follows a stable no-regret algorithm and $n$ is the dimension of the row player's strategy. Then, by following LRCA,  for any $\epsilon >0$, there exists $l \in \mathbb{N}$ such that $\frac{\mathcal{R}_{l}}{l}=\mathcal{O}(\frac{\epsilon^2}{n})$ and $f(\mathbf{x}_l)- v \leq \epsilon$.
\end{theorem}

\begin{proof}
We will prove the theorem by contradiction. Suppose there exists $\epsilon > 0$ such that: 
\[ f(\mathbf{x}_l) - v > \epsilon, \; \forall l \in \mathbb{N}.
\]
Then, following the update rule of Algorithm 1 (LRCA) we have:
\[\mathbf{y}_{2k-1}=\mathbf{y}^* \; ; \; \alpha_{2k} = \frac{ f(\mathbf{x}_{2k-1}) - v}{\max(\frac{n}{4}, 2)} > \frac{\epsilon}{\max(\frac{n}{4}, 2)}.
\]
By the stability property, as $\mathbf{y}_{2k-1} = \mathbf{y}^*$, we then have: $\mathbf{x}_{2k-1}= \mathbf{x}_{2k}$. Following the update rule of Algorithm 1 (LRCA):
\begin{subequations}
    \begin{align}
     \mathbf{x}_{2k}^{T} A\mathbf{y}_{2k} &= \mathbf{x}_{2k-1}^{T}  A \left( (1-\alpha_{2k})\mathbf{y}^*+ \alpha_{2k}\mathbf{e}_{2k} \right)\nonumber \\
     &\geq (1-\alpha_{2k}) v +\alpha_{2k} f(\mathbf{x}_{2k-1}) \label{Stablity 1 a} \\
     &> (1-\alpha_{2k}) v +\alpha_{2k} (v+\epsilon) \label{Stablity 1 b} \\
     &\geq v + \frac{\epsilon^2}{\max (\frac{n}{4},2 )} \nonumber,
    \end{align}
\end{subequations}
Where inequality (\ref{Stablity 1 a}) is due to 
\[\mathbf{x}^{T} A\mathbf{y}^* \geq v \; \forall \mathbf{x} \in \Delta_n,\]
and where inequality (\ref{Stablity 1 b}) comes from the assumption that $f(\mathbf{x}_l)- v > \epsilon$.
We then have: 
\[\frac{1}{T} \sum_{t=1}^T  \mathbf{x}_t^{T}A\mathbf{y}_t \geq \frac{v+ \left( v + \frac{\epsilon^2}{\max (\frac{n}{4},2 )}\right)}{2} = v+ \frac{\epsilon^2}{2\max (\frac{n}{4},2 )} .\]
We also note that, from the definition of the value of the game, we have: 
\[\min_{i}\frac{1}{T} \sum_{t=1}^T  \mathbf{e}_i^{T} A\mathbf{y}_t = \min_{i}\mathbf{e}_i^{T} A  \frac{\sum_{t=1}^T \mathbf{y}_t}{T} \leq v.
\]
Thus, we have: 
\[
\lim_{T \rightarrow \infty} {\min_{i}\frac{1}{T} \sum_{t=1}^T  \mathbf{e}_i^{T} A\mathbf{y}_t-\frac{1}{T} \sum_{t=1}^T  \mathbf{x}_t^{T} A\mathbf{y}_t} \leq v - \left( v+ \frac{\epsilon^2}{2\max (\frac{n}{4},2 )} \right) = - \frac{\epsilon^2}{2\max (\frac{n}{4},2 )},  \]
contradicting to the definition of a no-regret algorithm: 
\[
\lim_{T \rightarrow \infty} {\min_{i}\frac{1}{T} \sum_{t=1}^T  \mathbf{e}_i^{T} A\mathbf{y}_t-\frac{1}{T} \sum_{t=1}^T  \mathbf{x}_t^{T} A\mathbf{y}_t}=0. \]
\end{proof}

Note that $(\mathbf{x}_l, \mathbf{y}^*)$ which satisfy $f(\mathbf{x}_l)-v \leq \epsilon$ are $\epsilon$-Nash equilibria and $\epsilon=0$ implies $\mathbf{x}_l$ is the Nash equilibrium of the row player. For no-regret algorithms with optimal regret bound $\mathcal{R}_{l} = O(\sqrt{l})$, following Theorem \ref{extension of LRCA algorithm}, the row player will reach an $\epsilon$-Nash equilibrium in at most $\mathcal{O}(\frac{1}{\epsilon})^4$ rounds. Due to the full information feedback assumption, the column player will know when the row player plays an $\epsilon$-Nash equilibrium strategy. Depending on the number of rounds, the column player can lead the row player to play any $\epsilon$-Nash equilibrium, after that switching from LRCA to $\mathbf{y}^*$ so the row player will remain playing the $\epsilon$-Nash equilibrium (due to the stability property).

Note that, for every $\delta >0$, there exists an $\epsilon >0$ such that every $\epsilon$-Nash equilibria will lie within a $\delta$-neighbourhood of some minimax strategy for the row player. Since $\mathbf{x}^{*}$ is the only minimax strategy for the row player, we can pick a sufficiently small $\epsilon$ such that any $\epsilon$-Nash equilibrium strategy must be within a $\delta$-neighbourhood of $\mathbf{x}^{*}$ for some $\delta > 0$. Thus, by guiding the row player towards an $\epsilon$-Nash equilibrium strategy, the LRCA algorithm has guided the row player to a strategy within a $\delta$-neighbourhood of $\mathbf{x}^{*}$, as required.

Note that there are many no-regret algorithms with stability property, 
For example, we prove below that a wide class of no-regret algorithms, Follow The Regularized Leader (FTRL), are stable:
\begin{theorem} \label{FTRL has stability property}
Suppose the row player follows a FTRL algorithm with regularizer  $R(x)$ defined as:
\[x_t = \argmin_{x \in \Delta_n} x^\top  \left(\sum_{i=1}^{t-1} Ay_i\right) + R(x).\]
If there exists a fully-mixed minimax equilibrium strategy for the row player, then FTRL is stable.
\end{theorem}
\begin{proof}
As there exists a fully-mixed minimax equilibrium strategy for the row player, we have $A\mathbf{y}^*=\left[v, \dots, v\right]^{T}$. Thus, we have:
\[\mathbf{x}^TA\mathbf{y}^*=v \quad \forall \mathbf{x} \in \Delta_n\]
When the column player follows the minimax strategy, the minimization for $\mathbf{x}_t$ and $\mathbf{x}_{t+1}$ only differ by a constant term $v$, so their solutions are the same.
\end{proof}
Note that the FTRL framework, with appropriately chosen $R(x)$, can recover many popular no-regret algorithms, such as online mirror descent, multiplicative weights update, and Hannan's algorithm (a.k.a. Follow the Perturbed Leader). See, e.g., ~\cite{mcmahan2011follow,Arora2012} for more details.
\section{Conclusion}
In this paper, we have studied a repeated two-player zero-sum game setting in which, the column player has control over the payoff matrix and aims to guide the row player into exhibiting a specific behaviour in the form a mixed strategy. We developed a novel method for constructing a payoff matrix which guarantees the only minimax strategy available to the row player is the mixed strategy of interest to the system designer. This method was based on simultaneously satisfying three systems of inequalities, two of which were linear. Using such a payoff matrix, we showed that the LRCA no-regret learning algorithm can be used to effectively guide the row player within a $\delta$-neighbourhood of the strategy desired. Moreover, we showed that the last iterate convergence properties of the LRCA algorithm apply to a much wider class of no-regret algorithms than previously acknowledged. This class of algorithms is described the by stability property, introduced and defined in Section \ref{section: last round convergence}. 

One disadvantage of our approach is that the system designer requires full control over the payoff matrix. In the future, we will consider a setting in which the column player has limited control over the payoff matrix. In reality, there may also exist a ground truth payoff matrix, which the system designer must modify at some cost. In this case, choosing a suitable payoff matrix forms a difficult optimisation problem, and it is open question whether such a problem can be solved efficiently.


\bibliographystyle{plain}
\bibliography{main}
\end{document}